\documentclass[12pt]{article}

\usepackage[usenames]{xcolor}
\usepackage{graphicx}
\usepackage[colorlinks=true,
linkcolor=webgreen,
filecolor=webbrown,
citecolor=webgreen]{hyperref}
\definecolor{webgreen}{rgb}{0,.5,0}
\definecolor{webbrown}{rgb}{.6,0,0}
\usepackage{color}

\usepackage{amsthm}
\usepackage{mathtools}

\usepackage{tikz}
\usepackage{subcaption}
\usepackage{verbatim}
\usepackage{diagbox}
\usepackage{float}

\theoremstyle{plain}
\newtheorem{theorem}{Theorem}
\newtheorem{corollary}[theorem]{Corollary}
\newtheorem{lemma}[theorem]{Lemma}
\newtheorem{proposition}[theorem]{Proposition}

\theoremstyle{definition}

\theoremstyle{remark}
\newtheorem{remark}[theorem]{Remark}

\newcommand{\floor}[1]{\left\lfloor#1\right\rfloor}
\newcommand{\ceil}[1]{\left\lceil#1\right\rceil}
\newcommand{\seqnum}[1]{\href{https://oeis.org/#1}{\underline{#1}}}

\begin{document}

\begin{center}
\vskip 1cm{\LARGE\bf 
Computationally Inequivalent Summations and Their Parenthetic Forms}
\vskip 1cm
\large
Laura Monroe and Vanessa Job \\
Ultrascale Systems Research Center\\
Los Alamos National Laboratory\\
Los Alamos, NM 87501\\
USA \\
\href{mailto:lmonroe@lanl.gov}{\tt lmonroe@lanl.gov} \\
\href{mailto:vjob@lanl.gov}{\tt vjob@lanl.gov}
\end{center}

\vskip .2 in

\begin{abstract}
	Floating-point addition on a finite-precision machine is not associative, so not all mathematically equivalent summations are computationally equivalent. 
	Making this assumption can lead to numerical error in computations. 
	Proper ordering and parenthesizing is a low-overhead way of mitigating such  error in a floating point summation. 
	
	Ordered and parenthesized summations fall into equivalence classes. 
	We describe these classes, and the parenthetic forms summations in these classes take. 
	We provide summation-related interpretations for sequences known in other contexts, and give new recursive and closed formulas for sequences not previously related to summation. 
	
	We also introduce a data structure that facilitates understanding of these objects, and use it to consider certain forms of summation used by default in widely used computer languages. 
	Finally, we relate this data structure to other mathematical constructs from the fields of mathematical analysis and algorithmic analysis.
\end{abstract}

\section{Introduction}
Two summations are mathematically equivalent if they are comprised of the same summands, under any parenthesization and any ordering, since addition is associative and commutative. 
This is not true for floating-point addition on computers of finite precision \cite{goldberg91,higham93}. 

The assumption that mathematically equivalent summations are also computationally equivalent can lead to calculation error, and affect the accuracy of the result to the detriment of the overall calculation, since error in one step can propagate and influence the entire run. To complicate matters, there can easily be a very large number of non-equivalent summations, so it is not feasible to simply examine all possible cases. 

In this paper, we look at the structure of computationally equivalent forms of floating-point summations on $n$ term. We enumerate the non-equivalent groupings and orderings, and discuss in detail some cases that are of interest because of their high degree of accuracy or because they are default in commonly used programming languages. 

These  structures and this kind  of  analysis appear  in a  range of  disparate fields. For example, the form and  ordering of trees with similar structure is important in mathematical phylogenetics \cite{rosenberg19}. There is also  a  connection between the summation-based structures discussed here, the Karatsuba recursion tree \cite{karatsuba62, baruchel19, baruchel19_2}, and the Tagaki  function \cite{takagi01, lagarias11}.

\section{Background}

\subsection{Mathematically and computationally equivalent summations}
Two summations are \emph{computationally  equivalent} if they differ only by some series of pairwise transpositions. Two computationally equivalent summations will always have the same value on any system following IEEE 754 \cite{IEEE754} because the commutativity of addition is guaranteed under this standard. However, associativity is not guaranteed and should not be  assumed on a finite-precision machine.  

The equivalence relation in use throughout this paper is accordingly that imposed by IEEE 754: pairwise commutativity. All isomorphisms we discuss in this paper are in terms of this relation. When we call two summations equivalent in this paper, we mean computationally equivalent on a IEEE-754-compliant finite-precision machine. We do not mean mathematically equivalent, which of course does assume  associativity. 

Adding the same summands with different groupings and orderings that are computationally inequivalent can produce different results in practice, and can induce rounding error \cite{goldberg91,higham93}. This is discussed in detail in many references going back decades. 

As an example of computationally inequivalent summations, consider a summation  problem with one very large number and many very small numbers. Two groupings can lead to different results: 
\begin{itemize}
\item When adding a large floating-point number to a small one, the sum may exceed the precision of the computer, and the significant digits of the small number are lost. 
\item On the other hand, if many very small floats are added together first before adding the large number, the resulting sum of small elements may add to the end sum, even though some of the significant digits of the sum of small elements are still lost. 
\end{itemize}

\subsection{Grouping and ordering}
In this paper, we use the terms \emph{parenthesization} and  \emph{grouping} interchangeably.

Two  summations that  are computationally equivalent must have equivalent parenthetic form. However, having equivalent parenthetic form is not sufficient for computational equivalence. 

This reason for this is the importance of ordering. For any parenthetic form, one can always select a particular set of floating-point summands and impose an ordering in such a way that floating-point error takes place. On the other hand, optimal grouping combined with the best ordering for that parenthesization can lead to much-reduced rounding error \cite{job20}. For this reason, when we enumerate computationally inequivalent summations, we must consider both grouping and ordering. 

A simple example of two computationally  inequivalent summations that have  equivalent parenthetic form is $((a+b)+c)$
and $(a+(b+c))$. The sum $(a+(b+c))$ is computationally equivalent to $((b+c)+a)$, by pairwise commutativity, and $((b+c)+a)$ has the same parenthetic form as $((a+b)+c)$. However, the ordering of additions means that that $((b+c)+a)$ may not obtain the same results as $((a+b) + c$, and so they are not computationally equivalent.

\subsubsection{Implications for computation}
In some sense,  selecting a grouping is answering the question: ``In what manner should the summation take place?'', whereas selecting an ordering  answers: ``How shall we instantiate the chosen grouping?" 

Parenthesization is relevant to such things as parallelization, where a balanced tree is preferred, for good load-balancing. However, in a parallel implementation, there may not be control over the order in which the parallel sub-sums are added, and this must be  taken into  account when designing a parallelization scheme. Also, in a perfectly load-balanced scheme, one may find oneself obliged to group very large summands with very small, thus greatly enhancing the likelihood of computational error. 

Parenthesization also  comes into play when designing a default  summation form  for a given programming language. The specifications for the language may permit the programmer to sum values without defining the parenthesization. In that case, the compiler will address the parenthesization in some default manner that seems optimal.

On the  other hand, ordering influences the amount of rounding error one might see. There are parenthesizations that for certain distributions of summand values are unacceptably likely to give results prone to rounding error, as in the load-balanced scheme discussed above. 

One approach might be to start with a rough ordering, perhaps on subsets of the summands, select groupings that work for each  subset, and then add the  sums of the subset  sums in some appropriate grouping.

One might have computational constraints on the ordering of a summation rather than the grouping, for example on a streaming problem, but we do not address that question in this paper.

\subsection{Online Encyclopedia of Integer Sequences}
When a sequence in this paper associated with a class of summations  turns out to  exist in the Online Encyclopedia of Integer Sequences (OEIS) \cite{OEIS}, we name the OEIS sequence and cite the listing. This paper links several disparate OEIS sequences into a family of summation-related sequences.

The OEIS resource can be very useful in the investigation of mathematical objects associated with sequences. It lists other mathematical objects having the same sequences, and thus can give insight into the mathematics of the problem itself and its connection to other problems. It is well worth looking at any such references given, in the hope of finding commonalities with seemingly unrelated problems.

\section{Trees and summations}

\subsection{Binary tree representation of a summation}
A convenient and clarifying way of representing a summation is as a leaf-labeled rooted full binary tree. A full binary tree is one in which every node has either zero or two children \cite{knuth97art1}. In a summation tree, a node with two children represents the parenthesized sum of its children. The leaf nodes are labeled with the variables, and the internal nodes with +. 

In this paper, we use the terms \emph{summation} and \emph{leaf-labeled summation tree} interchangeably. We also use the terms \emph{parenthetic form} and \emph{unlabeled summation tree} interchangeably.
\subsubsection{Summations and their parenthetic forms}
We consider in this paper both abstract parenthetic forms and actual summations on these forms. 
Parenthetic forms are abstract, correspond to unlabeled summation trees, describe only the parenthesization, and may represent a class of summation trees, or a family of summations. 
Actual summations correspond to leaf-labeled trees, and describe both the parenthesization and the ordering of the variables. A summation is a concrete instantiation of some abstract parenthetic form.

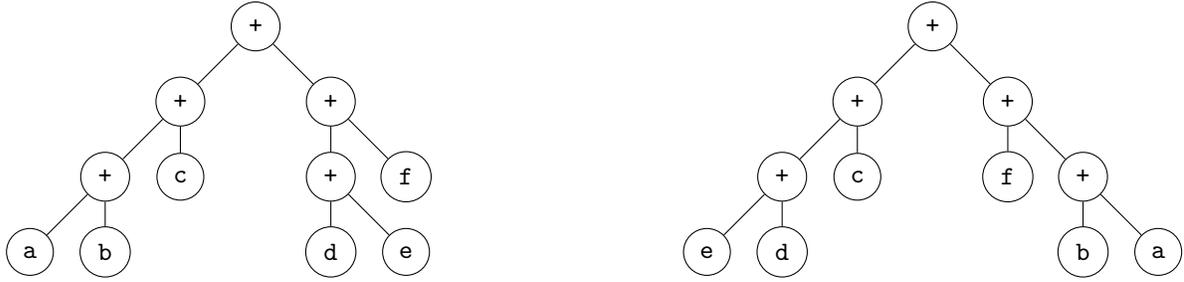
\begin{figure}[h!tb]
	\begin{center}\begin{tikzpicture}
		\tikzstyle{mynode}=[circle, draw]
		\tikzstyle{leaf}=[fill=white]
		\tikzstyle{indirect}=[pattern=north east lines, pattern color=gray!80]
		\tikzstyle{myarrow}=[]
		\node[mynode] (L1) at (0, 3) {\footnotesize\texttt{+}};
		\node[mynode] (L2a) at (-1, 2) {\footnotesize\texttt{+}};
		\node[mynode] (L2b) at (1, 2) {\footnotesize\texttt{+}};
		\node[mynode] (L3a) at (-2, 1) {\footnotesize\texttt{+}};
		\node[mynode,leaf] (L3b) at (-1, 1) {\footnotesize\texttt{c}};
		\node[mynode] (L3d) at (1, 1) {\footnotesize\texttt{+}};
		\node[mynode,leaf] (L3e) at (2, 1) {\footnotesize\texttt{f}};
		\node[mynode, leaf] (L4a) at (-3, 0) {\footnotesize\texttt{a}};
		\node[mynode, leaf] (L4b) at (-2, 0) {\footnotesize\texttt{b}};
		\node[mynode, leaf] (L4g) at (1, 0) {\footnotesize\texttt{d}};
		\node[mynode, leaf] (L4h) at (2, 0) {\footnotesize\texttt{e}};
		\draw[myarrow] (L1) to node[midway,above=0pt, left=7pt]{} (L2a);
		\draw[myarrow] (L1) to node[midway,above=0pt, left=-2.5pt]{} (L2b);
		\draw[myarrow] (L2a) to node[midway,above=0pt, left=4pt]{} (L3a);
		\draw[myarrow] (L2a) to node[midway,above=0pt, right=-1pt]{} (L3b);
		\draw[myarrow] (L2b) to node[midway,above=0pt, left=-0.5pt]{} (L3d);
		\draw[myarrow] (L2b) to node[midway,above=0pt, left=-1pt]{} (L3e);
		\draw[myarrow] (L3a) to node[midway,above=0pt, left=0mm]{} (L4a);
		\draw[myarrow] (L3a) to node[midway,above=0pt, right=-2pt]{} (L4b);
		\draw[myarrow] (L3d) to node[midway,above=0pt, left=0mm]{} (L4g);
		\draw[myarrow] (L3d) to node[midway,above=0pt, right=-2pt]{} (L4h);
		\tikzstyle{mynode}=[circle, draw]
		\tikzstyle{leaf}=[fill=white]
		\tikzstyle{indirect}=[pattern=north east lines, pattern color=gray!80]
		\tikzstyle{myarrow}=[]
		\node[mynode] (L1) at (9, 3) {\footnotesize\texttt{+}};
		\node[mynode] (L2a) at (8, 2) {\footnotesize\texttt{+}};
		\node[mynode] (L2b) at (10, 2) {\footnotesize\texttt{+}};
		\node[mynode] (L3a) at (7, 1) {\footnotesize\texttt{+}};
		\node[mynode,leaf] (L3b) at (8, 1) {\footnotesize\texttt{c}};
		\node[mynode, leaf] (L3d) at (10, 1) {\footnotesize\texttt{f}};
		\node[mynode] (L3e) at (11, 1) {\footnotesize\texttt{+}};
		\node[mynode, leaf] (L4a) at (6, 0) {\footnotesize\texttt{e}};
		\node[mynode, leaf] (L4b) at (7, 0) {\footnotesize\texttt{d}};
		\node[mynode, leaf] (L4g) at (11, 0) {\footnotesize\texttt{b}};
		\node[mynode, leaf] (L4h) at (12, 0) {\footnotesize\texttt{a}};
		\draw[myarrow] (L1) to node[midway,above=0pt, left=7pt]{} (L2a);
		\draw[myarrow] (L1) to node[midway,above=0pt, left=-2.5pt]{} (L2b);
		\draw[myarrow] (L2a) to node[midway,above=0pt, left=4pt]{} (L3a);
		\draw[myarrow] (L2a) to node[midway,above=0pt, right=-1pt]{} (L3b);
		\draw[myarrow] (L2b) to node[midway,above=0pt, left=-0.5pt]{} (L3d);
		\draw[myarrow] (L2b) to node[midway,above=0pt, left=-1pt]{} (L3e);
		\draw[myarrow] (L3a) to node[midway,above=0pt, left=0mm]{} (L4a);
		\draw[myarrow] (L3a) to node[midway,above=0pt, right=-2pt]{} (L4b);
		\draw[myarrow] (L3e) to node[midway,above=0pt, left=0mm]{} (L4g);
		\draw[myarrow] (L3e) to node[midway,above=0pt, right=-2pt]{} (L4h);
		\end{tikzpicture}\end{center}
	\caption{Two isomorphic  summation trees, or parenthetic forms, that represent  computationally inequivalent summations.  These  trees represent $(((a+b)+c)+((d+e)+f))$ and $(((e+d)+c)+(f+(b+a)))$, respectively. Although the summands are the same and the trees are isomorphic in form, the summations themselves are not computationally equivalent, because of the ordering of the summands: the  summation on the right cannot be derived from the summation on the left by a sequence of pairwise transpositions.}
	\label{sum_trees}
\end{figure}

Two  summations that  are computationally equivalent for  any choice of summands must have isomorphic summation trees. However, having isomorphic summation trees is not sufficient for computational equivalence. Figure \ref{sum_trees} is an example of two summations that have isomorphic unlabeled summation trees, yet have non-isomorphic leaf-labeled summation trees and are computationally inequivalent.

\subsubsection{Equivalent parenthetic forms and isomorphic unlabeled trees}
Two unlabeled summation trees are \emph{isomorphic} if one can be obtained from the other by a sequence of reversing the children of some set of nodes. This is in accordance with the pairwise commutativity of addition we assume. Parenthetic forms that are equivalent up to pairwise commutativity of summands map to isomorphic summation trees. 

This refers only to isomorphic equivalence of parenthetic form; it is \emph{not} the same as computational equivalence on instantiated summations, which depends on ordering as well as grouping. 

\subsubsection{Computationally equivalent summations and isomorphic leaf-labeled trees}
Two leaf-labeled summation trees are \emph{isomorphic} if one can be obtained from the other by a sequence of reversing the children of some set of nodes. Again, this is in accordance with the pairwise commutativity of addition. Computationally equivalent summations map to isomorphic leaf-labeled trees.

\subsection{SD-trees and enumerations}
A summation tree can be considered as a type of PQ tree. A PQ tree is a labeled tree representing a set of permutations on $n$ elements, where the $n$ leaf nodes are labeled with the elements and the internal nodes are labeled as $P$ or $Q$, where the $P$ nodes have at least two children with all permutations of the children  equivalent, and the $Q$ nodes have at least three children, with reversals of the orderings of the children equivalent \cite{booth76}. A summation tree is a binary PQ-tree, and since it is binary, all of the internal nodes are labeled with $P$. 

To aid in the enumeration and understanding of summation trees, we introduce a type of PQ summation tree called an SD-tree. We use this structure to establish a  formula and a general method for enumerating members of a class of summations.

\subsubsection{SD-trees}
An \emph{SD-tree} is a binary tree in which all internal nodes have two children, and where an internal node is labeled $S$ if its two children have the \{S\}ame number of descendant leaf nodes, and $D$ if its two children have \{D\}ifferent numbers of descendant leaf nodes. One might also consider a leaf-labeled SD-tree, which is an instantiation of the parenthetic form represented by the SD-tree in question, and represents an actual summation. An SD-tree is a variant of a PQ-tree. An example of a leaf-labeled SD-tree is shown in figure \ref{sd_tree}.

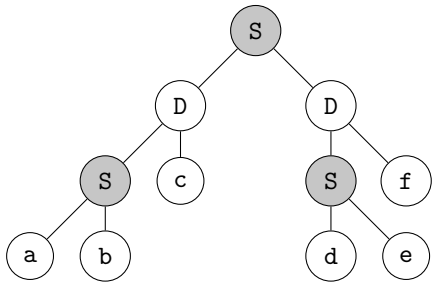
\begin{figure}[h!tb]
	\begin{center}\begin{tikzpicture}
		\tikzstyle{mynode}=[circle, draw]
		\tikzstyle{leaf}=[fill=white]
		\tikzstyle{Snode}=[fill=gray!45]
		\tikzstyle{indirect}=[pattern=north east lines, pattern color=gray!80]
		\tikzstyle{myarrow}=[]
		\node[mynode, Snode] (L1) at (0, 3) {\footnotesize\texttt{S}};
		\node[mynode] (L2a) at (-1, 2) {\footnotesize\texttt{D}};
		\node[mynode] (L2b) at (1, 2) {\footnotesize\texttt{D}};
		\node[mynode, Snode] (L3a) at (-2, 1) {\footnotesize\texttt{S}};
		\node[mynode,leaf] (L3b) at (-1, 1) {\footnotesize\texttt{c}};
		\node[mynode, Snode] (L3d) at (1, 1) {\footnotesize\texttt{S}};
		\node[mynode,leaf] (L3e) at (2, 1) {\footnotesize\texttt{f}};
		\node[mynode, leaf] (L4a) at (-3, 0) {\footnotesize\texttt{a}};
		\node[mynode, leaf] (L4b) at (-2, 0) {\footnotesize\texttt{b}};
		\node[mynode, leaf] (L4g) at (1, 0) {\footnotesize\texttt{d}};
		\node[mynode, leaf] (L4h) at (2, 0) {\footnotesize\texttt{e}};
		\draw[myarrow] (L1) to node[midway,above=0pt, left=7pt]{} (L2a);
		\draw[myarrow] (L1) to node[midway,above=0pt, left=-2.5pt]{} (L2b);
		\draw[myarrow] (L2a) to node[midway,above=0pt, left=4pt]{} (L3a);
		\draw[myarrow] (L2a) to node[midway,above=0pt, right=-1pt]{} (L3b);
		\draw[myarrow] (L2b) to node[midway,above=0pt, left=-0.5pt]{} (L3d);
		\draw[myarrow] (L2b) to node[midway,above=0pt, left=-1pt]{} (L3e);
		\draw[myarrow] (L3a) to node[midway,above=0pt, left=0mm]{} (L4a);
		\draw[myarrow] (L3a) to node[midway,above=0pt, right=-2pt]{} (L4b);
		\draw[myarrow] (L3d) to node[midway,above=0pt, left=0mm]{} (L4g);
		\draw[myarrow] (L3d) to node[midway,above=0pt, right=-2pt]{} (L4h);
		\end{tikzpicture}\end{center}
	\caption{An example  of an SD-tree. This tree represents $(((a+b)+c)+((d+e)+f))$. This is the same summation as in the first tree in Figure \ref{sum_trees}. The nodes labeled $S$ have children with the same number of descendant leaves. The nodes labeled $D$ have children with different numbers of descendant leaves.}
	\label{sd_tree}		
\end{figure}

Two summation SD-trees are  considered isomorphic if one can be obtained from the other by a sequence of reversing the children of some set of nodes (this is the same definition  as in the non-SD summation tree case).  
\begin{lemma}\label{iso_SD}
	Two isomorphic SD-trees with $n$ leaf nodes have the  same number of $S$-nodes and the same number of $D$-nodes.
\end{lemma}
\begin{proof}
	Follows from pairwise commutativity of addition.
\end{proof}
\begin{remark}
	The converse is not true. Two SD-trees with $n$ leaf nodes and with the same number  of $S$-nodes and $D$-nodes need not be isomorphic, as in Figure \ref{trees_noniso}.
\end{remark}

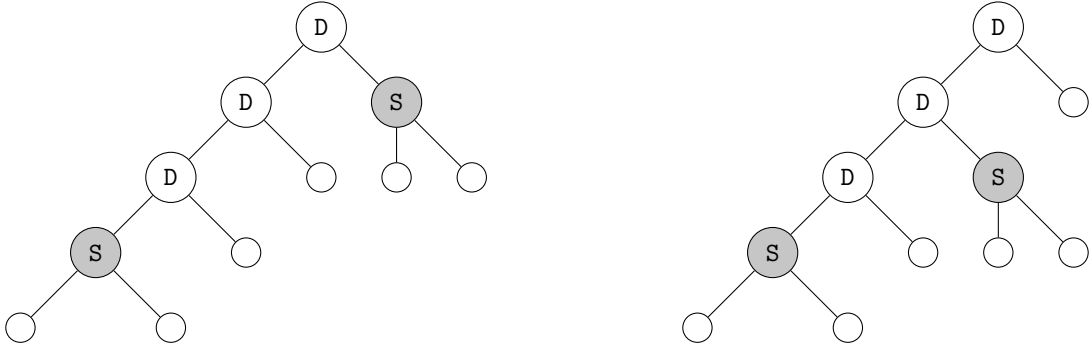
\begin{figure}[h!tb]
	\begin{center}\begin{tikzpicture}
		\tikzstyle{mynode}=[circle, draw]
		\tikzstyle{leaf}=[fill=gray!45]
		\tikzstyle{Snode}=[fill=gray!45]
		\tikzstyle{indirect}=[pattern=north east lines, pattern color=gray!80]
		\tikzstyle{myarrow}=[]
		\node[mynode] (L1) at (0, 3) {\footnotesize\texttt{D}};
		\node[mynode] (L2a) at (-1, 2) {\footnotesize\texttt{D}};
		\node[mynode,Snode] (L2b) at (1, 2) {\footnotesize\texttt{S}};
		\node[mynode] (L3a) at (-2, 1) {\footnotesize\texttt{D}};
		\node[mynode] (L3b) at (0, 1) {\footnotesize\texttt{}};
		\node[mynode] (L3d) at (1, 1) {\footnotesize\texttt{}};
		\node[mynode] (L3e) at (2, 1) {\footnotesize\texttt{}};
		\node[mynode,Snode] (L4a) at (-3, 0) {\footnotesize\texttt{S}};
		\node[mynode] (L4b) at (-1, 0) {\footnotesize\texttt{}};
		\node[mynode] (L5a) at (-4, -1) {\footnotesize\texttt{}};
		\node[mynode] (L5b) at (-2, -1) {\footnotesize\texttt{}};
		\draw[myarrow] (L1) to node[midway,above=0pt, left=7pt]{} (L2a);
		\draw[myarrow] (L1) to node[midway,above=0pt, left=-2.5pt]{} (L2b);
		\draw[myarrow] (L2a) to node[midway,above=0pt, left=4pt]{} (L3a);
		\draw[myarrow] (L2a) to node[midway,above=0pt, right=-1pt]{} (L3b);
		\draw[myarrow] (L2b) to node[midway,above=0pt, left=-0.5pt]{} (L3d);
		\draw[myarrow] (L2b) to node[midway,above=0pt, left=-1pt]{} (L3e);
		\draw[myarrow] (L3a) to node[midway,above=0pt, left=0mm]{} (L4a);
		\draw[myarrow] (L3a) to node[midway,above=0pt, right=-2pt]{} (L4b);
		\draw[myarrow] (L4a) to node[midway,above=0pt, left=0mm]{} (L5a);
		\draw[myarrow] (L4a) to node[midway,above=0pt, right=-2pt]{} (L5b);
		\tikzstyle{mynode}=[circle, draw]
		\tikzstyle{leaf}=[fill=gray!45]
		\tikzstyle{indirect}=[pattern=north east lines, pattern color=gray!80]
		\tikzstyle{myarrow}=[]
		\node[mynode] (L1) at (9, 3) {\footnotesize\texttt{D}};
		\node[mynode] (L2a) at (8, 2) {\footnotesize\texttt{D}};
		\node[mynode] (L2b) at (10, 2) {\footnotesize\texttt{}};
		\node[mynode] (L3a) at (7, 1) {\footnotesize\texttt{D}};
		\node[mynode,Snode] (L3b) at (9, 1) {\footnotesize\texttt{S}};
		\node[mynode,Snode] (L4a) at (6, 0) {\footnotesize\texttt{S}};
		\node[mynode] (L4b) at (8, 0) {\footnotesize\texttt{}};
		\node[mynode] (L4c) at (9, 0) {\footnotesize\texttt{}};
		\node[mynode] (L4d) at (10, 0) {\footnotesize\texttt{}};
		\node[mynode] (L5a) at (5, -1) {\footnotesize\texttt{}};
		\node[mynode] (L5b) at (7, -1) {\footnotesize\texttt{}};
		\draw[myarrow] (L1) to node[midway,above=0pt, left=7pt]{} (L2a);
		\draw[myarrow] (L1) to node[midway,above=0pt, left=-2.5pt]{} (L2b);
		\draw[myarrow] (L2a) to node[midway,above=0pt, left=4pt]{} (L3a);
		\draw[myarrow] (L2a) to node[midway,above=0pt, right=-1pt]{} (L3b);
		\draw[myarrow] (L3a) to node[midway,above=0pt, left=0mm]{} (L4a);
		\draw[myarrow] (L3a) to node[midway,above=0pt, right=-2pt]{} (L4b);
		\draw[myarrow] (L3b) to node[midway,above=0pt, left=0mm]{} (L4c);
		\draw[myarrow] (L3b) to node[midway,above=0pt, right=-2pt]{} (L4d);
		\draw[myarrow] (L4a) to node[midway,above=0pt, left=0mm]{} (L5a);
		\draw[myarrow] (L4a) to node[midway,above=0pt, right=-2pt]{} (L5b);
	\end{tikzpicture}\end{center}
	\caption{Not all trees with the same number of leaves and the same number of $S$-nodes are isomorphic. Here are  two non-isomorphic SD-trees with $6$ leaves,  $2$ $S$-nodes  and $3$ $D$-nodes. This example is smallest in terms of the number of leaves of the SD-tree. (There is also another example of non-isomorphic SD-trees with $6$ leaves, this having $3$ $S$-nodes and $2$ $D$-nodes.)}
	\label{trees_noniso}
\end{figure}

\subsection{Enumeration of summations using SD-trees}\label{method}
If the number of S-nodes in the elements of an equivalence class of summation SD-trees is known or can be quantified, then there is a simple formula to enumerate the computationally inequivalent summations in that class.
\begin{lemma}
	Let $T$ be SD-tree with $n$ leaf nodes. The number of computationally inequivalent leaf-labeled  SD-trees that are  isomorphic to T is  $\frac{n!}{2^\varepsilon}$, where $\varepsilon$ is the number of internal nodes of $T$ labeled S, i.e., with left and right subtrees having the same number of leaf nodes. 
\end{lemma}
\begin{proof}
	There are $n!$ ways to order the $n$ leaf nodes of the tree $T$.  Without loss of generality, it suffices to consider SD-trees where the number of leaf descendants of the left child of a node is always less than or equal to the number of leaf descendants of the right child. (This is because of pairwise commutativity.) 
	
	Call the set of left descendant leaf nodes of a node $d_\ell$, and the set of right descendant leaf nodes $d_r$. $|d_\ell|=|d_r|$ if and only if an ordering transposing $d_\ell$ and $d_r$ gives a tree that is the same as the original tree, up  to transposition of  the two subtrees. So  such subtrees are combinatorially counted twice. The number of trees that are the same as another given ordering is $2^\varepsilon$, so the total number of non-equivalent trees isomorphic to T is $\frac{n!}{2^\varepsilon}$.
\end{proof}
\begin{corollary}\label{equiv_sum}
    Let $T$ be  an SD-tree with $n$ leaf nodes. The number of non-equivalent summations corresponding (up to tree isomorphism) to $T$ is $\sigma(n)=\frac{n!}{2^\varepsilon}$, where $\varepsilon$ is the number of internal nodes of $T_{\sigma(n)}$ labeled S, i.e., with left and right subtrees having the same number of leaf nodes. 
\end{corollary}

\subsection{A general method for enumeration of summations}
Corollary \ref{equiv_sum} implies that it suffices to (1) show that all summation trees representing class members are isomorphic to each other as SD-trees, 
(2)  show that all summations with trees isomorphic to that class are in the class, and (3) provide a formula for the number of $S$ nodes in the representative summation tree. The enumeration then follows immediately. We will follow this method throughout this paper to enumerate summations of certain interesting classes.

\subsection{Parenthetic constraints on summations}
Certain classes of summation problems have  constraints placed on the grouping that may induce  isomorphic parenthetic forms. When such constraints exist, we can discuss the problem in terms of the summation trees that have forms conforming to the constraints. 
In other words, we fix the parenthesization in accordance with the constraint, then set the ordering. 

In the next two sections, we first discuss summations on $n$ variables, grouped and ordered without any constraints on the grouping. We then turn to summations with grouping constraints that induce specific parenthetic forms, and discuss in depth two important cases.

\section{Summations with no grouping constraints}

We assume in this section that all summations are completely and explicitly parenthesized, and that pairwise commutativity holds. We make no other assumptions; in other words, there are no constraints imposed on the parenthesizations.

The following formula is known to apply to leaf-labeled binary trees, and is shown by Stanley \cite{stanley97},  Callan \cite{callan09}  and Dale \cite{dale93}. Walters makes this observation on commutative, non-associative multiplication, which has the same  properties needed for this purpose as addition \cite{OEISA001147}.

We observe that the number of computationally inequivalent summations on $n$ variables is the same as the number of leaf-labeled rooted binary trees with $n$ leaf nodes. We give a sketch of a tree-based proof here, and also provide a new combinatorial proof.
\begin{proposition}\label{ineq_sum}\cite{stanley97, callan09, dale93}
	The number of  computationally inequivalent summations on $n$ terms is
	\begin{equation}
	(2n-3)!!
	\end{equation} 	
\end{proposition}
\begin{proof}[Proof 1, Tree-based Proof (sketch).] 	
	We proceed via induction, and note that all $n$-variable sums are obtained by taking the appropriate $(n-1)$-variable sum and adding the $n^\text{th}$ variable to one of its sub-sums. There are $(n-2)$ + signs in an $(n-1)$-variable sum, so the number of sub-sums to either side of a + sign is $2((n-2)-1)$. Including the expression itself gives at least $2((n-2)-1)+1 = (2n-3)$ possible sub-sums. Since there is no overlap, there is no duplication, and this is the exact number. So by the induction hypothesis, the number of computationally inequivalent summations on $n$ variables is $(2n-3)(2(n-1)-3)!!=(2n-3)(2n-5)!!=(2n-3)!!$.
\end{proof}	
\begin{proof}[Proof 2, Combinatorial Proof] 	
	We can order the $n$ variables in $n!$ ways to make $n!$ unparenthesized sums. 
	For each of these sums, there are $C_{n-1}= \frac{1}{n}\binom{2(n-1)}{n-1}$ distinct ways to parenthesize them,
	where $C_{n-1}$ is the $n^\text{th}$ Catalan number \cite{stanley97, OEISA000108}. 
	
	Each sum contains $n-1$ addition signs, and since the addends on either side of an addition sign can appear in two different orientations, each expression is equivalent to $2^{n-1}$ expressions under commutativity.  Thus, there are 
	\begin{align*}
	\frac{n! \cdot \frac{1}{n} \cdot 
	\binom{2(n-1)} {n-1}  }   {2^{n-1}}
	&=  \frac {(n-1)!   (2(n-1))! } {(n-1)! (n-1)! \cdot 2^{n-1}}  \\
	&= \frac {(2n-2)!} 
	{(n-1)! \cdot 2^{n-1}}\\
	&= \frac {(2n-2)!} 
	{ (2n-2) \cdot (2n-4) \cdots 4 \cdot 2} \\
	&=  \frac {(2n-2) (2n-3)(2n-4) \cdots 5 \cdot 4 \cdot 3 \cdot 2 \cdot 1} 
	{ (2n-2) \cdot (2n-4) \cdots 4 \cdot 2} \\
	&= (2n-3)!!
	\end{align*}	
\end{proof}
\begin{remark}
	This sequence in Proposition \ref{ineq_sum} is OEIS A001147 \cite{OEISA001147}. There are many different interpretations of this sequence. 
\end{remark}

\section{Summations with grouping constraints}
We discuss pairwise and ladder summation in depth in this section, two important cases that are used by default in certain widely used programming languages. Both of these have grouping constraints and both are examples of classes that can be described in terms of isomorphic SD-trees.

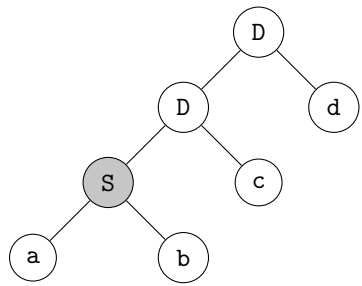
\begin{figure}[h!tb]
	\begin{center}\begin{tikzpicture}
	\tikzstyle{mynode}=[circle, draw]
	\tikzstyle{leaf}=[fill=white]
	\tikzstyle{Snode}=[fill=gray!45]
	\tikzstyle{indirect}=[pattern=north east lines, pattern color=gray!80]
	\tikzstyle{myarrow}=[]
	\node[mynode] (L1) at (0, 3) {\footnotesize\texttt{D}};
	\node[mynode] (L2a) at (-1, 2) {\footnotesize\texttt{D}};
	\node[mynode,leaf] (L2b) at (1, 2) {\footnotesize\texttt{d}};
	\node[mynode,Snode] (L3a) at (-2, 1) {\footnotesize\texttt{S}};
	\node[mynode,leaf] (L3b) at (0, 1) {\footnotesize\texttt{c}};
	\node[mynode, leaf] (L4a) at (-3, 0) {\footnotesize\texttt{a}};
	\node[mynode, leaf] (L4b) at (-1, 0) {\footnotesize\texttt{b}};
	\draw[myarrow] (L1) to node[midway,above=0pt, left=7pt]{} (L2a);
	\draw[myarrow] (L1) to node[midway,above=0pt, left=-2.5pt]{} (L2b);
	\draw[myarrow] (L2a) to node[midway,above=0pt, left=4pt]{} (L3a);
	\draw[myarrow] (L2a) to node[midway,above=0pt, right=-1pt]{} (L3b);
	\draw[myarrow] (L3a) to node[midway,above=0pt, left=0mm]{} (L4a);
	\draw[myarrow] (L3a) to node[midway,above=0pt, right=-2pt]{} (L4b);
	\end{tikzpicture}\end{center}
	\caption{An  example  of a ladder,  or serial, SD-tree, representing $(((a+b)+c)+d)$.}
	\label{ladder_tree}		
\end{figure}

\subsection{Ladder (serial) summation}
A ladder summation is one in which the summation proceeds pairwise in the order in which the terms appear. For example, the ladder summation on four summands is $(((a+b)+c)+d)$, and its tree is shown in Figure \ref{ladder_tree}. Ladder summation corresponds to the C language default of left-to-right associativity on summations with ungrouped summands \cite{C18}.
\begin{remark}
	By pairwise commutativity, the same result is guaranteed in C on an IEEE-754-compliant system upon transposition of the first two elements of an ungrouped summation,
	but is not guaranteed after any other transposition. 
\end{remark}
\begin{lemma}
	There is a unique ladder SD-tree with $n$ leaf nodes, and it has exactly one $S$-node.
\end{lemma}
\begin{proof}
	The lowest internal level of the a ladder-summation tree consists of one node with two leaf children. Every other node in the tree has one child that is a leaf; the other is the root of another ladder summation. So only the unique lowest internal node is labeled S.
\end{proof}
\begin{lemma}
	All SD-trees having exactly one $S$-node are isomorphic to some ladder SD-tree.
\end{lemma}
\begin{proof}
	Any subtree of an SD-tree must have at least one $S$-node, except when it consists of a single leaf node. This means one of the top branches of an SD-tree with exactly one $S$-node is a single leaf node. The conclusion follows by induction on the number of leaf nodes.
\end{proof}
\begin{corollary} \label{laddercor}
	Ladder SD-trees are exactly those with one $S$-node, up to isomorphism.
\end{corollary}
\begin{proposition}\label{numladder}
	The number of computationally inequivalent ladder summations on $n$ variables is 
	\begin{equation}
		\frac{n!}{2}
	\end{equation}
\end{proposition}
\begin{proof}
	Follows immediately from Corollary \ref{equiv_sum} and Corollary \ref{laddercor}.
\end{proof}
\begin{remark}  
	This is OEIS A001710 \cite{OEISA001710}, the number of even permutations on $n$ letters.
\end{remark}

\subsection{Pairwise summation}
Pairwise summation is a recursive summation method that proceeds by dividing the set of $n$ summands in half (or almost in half, if $n$ is odd), summing on the two subsets, then adding the two sums. Pairwise summation is known to be fairly accurate, and in some cases is nearly as accurate as such gold-standard techniques as Kahan summation \cite{higham93}. Pairwise summation is the default on ungrouped summands in NumPy \cite{numpysum} and in Julia \cite{julia}.

We discuss here three different forms for the sequence $\sigma(n)$, the number of pairwise summations on $n$ variables. Two are recursive and one is a closed form. One of the recursive  formulations is previously known, and is shown by David in \cite{david88}. 
\begin{proposition}\label{pairrecursiveold} \cite{david88}
	The number of computationally inequivalent pairwise summations on $n$ variables is 
	\begin{equation}\label{recursive1_A096351}
		\sigma(n)=
		\begin{cases}
			\frac{1}{2}\binom{2m} {m}\sigma(m)^2, & \text{if}\ n=2m; \\
			\binom{2m+1}{m}\sigma(m)\sigma(m+1), & \text{if}\ n=2m+1.
		\end{cases}
	\end{equation}	 
\end{proposition}
\begin{remark}
	Enumerating the pairwise summations on $n$ elements is equivalent to enumerating the tournaments on $n$ teams. This is the classical formulation of this problem, and gives rise to the sequence OEIS A096351  \cite{OEISA096351}. 
\end{remark}

\begin{figure}[h!tb]
	\tikzstyle{mynode}=[circle, draw]
	\tikzstyle{rectanglenode}=[rectangle, draw]
	\tikzstyle{Snode}=[fill=gray!45]
	\tikzstyle{leaf}=[fill=gray!45]
	\tikzstyle{indirect}=[pattern=north east lines, pattern color=gray!80]
	\tikzstyle{myarrow}=[]
	\centering
	\begin{subfigure}[t]{.36\textwidth}
		\begin{tikzpicture}
			\centering
			\node[mynode,Snode] (L1) at (4.5, 11) {\footnotesize\texttt{S}};
			\node[mynode,Snode] (L2a) at (3, 10) {\footnotesize\texttt{S}};
			\node[mynode,Snode] (L2b) at (6, 10) {\footnotesize\texttt{S}};
			\node[mynode] (L3a) at (2, 9) {\footnotesize\texttt{a}};
			\node[mynode] (L3d) at (5.5, 9) {\footnotesize\texttt{e}};
			\node[mynode] (L3b) at (3.5, 9) {\footnotesize\texttt{c}};
			\node[mynode] (L3e) at (7, 9) {\footnotesize\texttt{g}};
			\draw[myarrow] (L1) to node[midway,above=0pt, left=7pt]{} (L2a);
			\draw[myarrow] (L1) to node[midway,above=0pt, left=-2.5pt]{} (L2b);
			\draw[myarrow] (L2a) to node[midway,above=0pt, left=4pt]{} (L3a);
			\draw[myarrow] (L2a) to node[midway,above=0pt, right=-1pt]{} (L3b);
			\draw[myarrow] (L2b) to node[midway,above=0pt, left=-0.5pt]{} (L3d);
			\draw[myarrow] (L2b) to node[midway,above=0pt, left=-1pt]{} (L3e);
		\end{tikzpicture}
    	\caption{Pairwise summation SD-tree for the set of 4 summands $\{a,c,e,g\}$. This tree has  3 $S$-nodes.}
	\end{subfigure}%
	
	\begin{subfigure}[t]{.45\textwidth}
		\begin{tikzpicture}	
			\centering
			\node[mynode] (L1) at (0, 7) {\footnotesize\texttt{D}};
			\node[mynode] (L2a) at (-1.5, 6) {\footnotesize\texttt{D}};
			\node[mynode,Snode] (L2b) at (1.5, 6) {\footnotesize\texttt{S}};
			\node[mynode,Snode] (L3a) at (-2.5, 5) {\footnotesize\texttt{S}};	

			\node[mynode] (L3d) at (1, 5) {\footnotesize\texttt{e}};
			\node[mynode] (L3b) at (-1, 5) {\footnotesize\texttt{c}};
			\node[mynode] (L3e) at (2.5, 5) {\footnotesize\texttt{g}};
			\node[mynode] (L4a) at (-3.5, 4) {\footnotesize\texttt{a}};
			\node[mynode] (L4b) at (-2.5, 4) {\footnotesize\texttt{b}};
			\draw[myarrow] (L1) to node[midway,above=0pt, left=7pt]{} (L2a);
			\draw[myarrow] (L1) to node[midway,above=0pt, left=-2.5pt]{} (L2b);
			\draw[myarrow] (L2a) to node[midway,above=0pt, left=4pt]{} (L3a);
			\draw[myarrow] (L2a) to node[midway,above=0pt, right=-1pt]{} (L3b);
			\draw[myarrow] (L2b) to node[midway,above=0pt, left=-0.5pt]{} (L3d);
			\draw[myarrow] (L2b) to node[midway,above=0pt, left=-1pt]{} (L3e);
			\draw[myarrow] (L3a) to node[midway,above=0pt, left=0mm]{} (L4a);
			\draw[myarrow] (L3a) to node[midway,above=0pt, right=-2pt]{} (L4b);
    	\end{tikzpicture}
    	\caption{Pairwise summation SD-tree for the set of 5 summands $\{a,b,c,e,g\}$. This tree has  2 $S$-nodes.}
	\end{subfigure}\hfill%
	\begin{subfigure}[t]{.5\textwidth}
		\begin{tikzpicture}	
			\centering
			\node[mynode,Snode] (L1) at (9, 7) {\footnotesize\texttt{S}};
			\node[mynode] (L2a) at (7.5, 6) {\footnotesize\texttt{D}};
			\node[mynode] (L2b) at (10.5, 6) {\footnotesize\texttt{D}};
			\node[mynode,Snode] (L3a) at (6.5, 5) {\footnotesize\texttt{S}};
			\node[mynode,Snode] (L3d) at (10, 5) {\footnotesize\texttt{S}};
			\node[mynode] (L3b) at (8, 5) {\footnotesize\texttt{c}};
			\node[mynode] (L3e) at (11.5, 5) {\footnotesize\texttt{g}};
			\node[mynode] (L4a) at (5.5, 4) {\footnotesize\texttt{a}};
			\node[mynode] (L4b) at (6.5, 4) {\footnotesize\texttt{b}};
			\node[mynode] (L4e) at (9.5, 4) {\footnotesize\texttt{e}};
			\node[mynode] (L4f) at (10.5, 4) {\footnotesize\texttt{f}};
			\draw[myarrow] (L1) to node[midway,above=0pt, left=7pt]{} (L2a);
			\draw[myarrow] (L1) to node[midway,above=0pt, left=-2.5pt]{} (L2b);
			\draw[myarrow] (L2a) to node[midway,above=0pt, left=4pt]{} (L3a);
			\draw[myarrow] (L2a) to node[midway,above=0pt, right=-1pt]{} (L3b);
			\draw[myarrow] (L2b) to node[midway,above=0pt, left=-0.5pt]{} (L3d);
			\draw[myarrow] (L2b) to node[midway,above=0pt, left=-1pt]{} (L3e);
			\draw[myarrow] (L3a) to node[midway,above=0pt, left=0mm]{} (L4a);
			\draw[myarrow] (L3a) to node[midway,above=0pt, right=-2pt]{} (L4b);
			\draw[myarrow] (L3d) to node[midway,above=0pt, left=0mm]{} (L4e);
			\draw[myarrow] (L3d) to node[midway,above=0pt, right=-2pt]{} (L4f);
    	\end{tikzpicture}
		\caption{Pairwise summation SD-tree for the set of 6 summands $\{a,b,c,e,f,g\}$. This tree has  3 $S$-nodes.}
	\end{subfigure}%

	\begin{subfigure}[t]{.45\textwidth}
		\begin{tikzpicture}	
			\centering
			\node[mynode] (L1) at (0, 3) {\footnotesize\texttt{D}};
			\node[mynode,Snode] (L2a) at (-1.5, 2) {\footnotesize\texttt{S}};
			\node[mynode] (L2b) at (1.5, 2) {\footnotesize\texttt{D}};
			\node[mynode,Snode] (L3a) at (-2.5, 1) {\footnotesize\texttt{S}};
			\node[mynode,Snode] (L3b) at (-1, 1) {\footnotesize\texttt{S}};
			\node[mynode,Snode] (L3d) at (1, 1) {\footnotesize\texttt{S}};
			\node[mynode] (L3e) at (2.5, 1) {\footnotesize\texttt{g}};
			\node[mynode] (L4a) at (-3.5, 0) {\footnotesize\texttt{a}};
			\node[mynode] (L4b) at (-2.5, 0) {\footnotesize\texttt{b}};
			\node[mynode] (L4c) at (-1.5, 0) {\footnotesize\texttt{c}};
			\node[mynode] (L4d) at (-0.5, 0) {\footnotesize\texttt{d}};
			\node[mynode] (L4e) at (0.5, 0) {\footnotesize\texttt{e}};
			\node[mynode] (L4f) at (1.5, 0) {\footnotesize\texttt{f}};
			\draw[myarrow] (L1) to node[midway,above=0pt, left=7pt]{} (L2a);
			\draw[myarrow] (L1) to node[midway,above=0pt, left=-2.5pt]{} (L2b);
			\draw[myarrow] (L2a) to node[midway,above=0pt, left=4pt]{} (L3a);
			\draw[myarrow] (L2a) to node[midway,above=0pt, right=-1pt]{} (L3b);
			\draw[myarrow] (L2b) to node[midway,above=0pt, left=-0.5pt]{} (L3d);
			\draw[myarrow] (L2b) to node[midway,above=0pt, left=-1pt]{} (L3e);
			\draw[myarrow] (L3a) to node[midway,above=0pt, left=0mm]{} (L4a);
			\draw[myarrow] (L3a) to node[midway,above=0pt, right=-2pt]{} (L4b);
			\draw[myarrow] (L3b) to node[midway,above=0pt, left=0mm]{} (L4c);
			\draw[myarrow] (L3b) to node[midway,above=0pt, right=-2pt]{} (L4d);
			\draw[myarrow] (L3d) to node[midway,above=0pt, left=0mm]{} (L4e);
			\draw[myarrow] (L3d) to node[midway,above=0pt, right=-2pt]{} (L4f);
		\end{tikzpicture}
		\caption{Pairwise summation SD-tree for the set of 7 summands $\{a,b,c,d,e,f,g\}$. This tree has  4 $S$-nodes.}
	\end{subfigure}\hfill%
	\begin{subfigure}[t]{.5\textwidth}
		\begin{tikzpicture}	
			\centering
			\node[mynode,Snode] (L1) at (9, 3) {\footnotesize\texttt{S}};
			\node[mynode,Snode] (L2a) at (7.5, 2) {\footnotesize\texttt{S}};
			\node[mynode,Snode] (L2b) at (10.5, 2) {\footnotesize\texttt{S}};
			\node[mynode,Snode] (L3a) at (6.5, 1) {\footnotesize\texttt{S}};
			\node[mynode,Snode] (L3b) at (8, 1) {\footnotesize\texttt{S}};
			\node[mynode,Snode] (L3d) at (10, 1) {\footnotesize\texttt{S}};
			\node[mynode,Snode] (L3e) at (11.5, 1) {\footnotesize\texttt{S}};
			\node[mynode] (L4a) at (5.5, 0) {\footnotesize\texttt{a}};
			\node[mynode] (L4b) at (6.5, 0) {\footnotesize\texttt{b}};
			\node[mynode] (L4c) at (7.5, 0) {\footnotesize\texttt{c}};
			\node[mynode] (L4d) at (8.5, 0) {\footnotesize\texttt{d}};
			\node[mynode] (L4e) at (9.5, 0) {\footnotesize\texttt{e}};
			\node[mynode] (L4f) at (10.5, 0) {\footnotesize\texttt{f}};
			\node[mynode] (L4g) at (11.5, 0) {\footnotesize\texttt{g}};
			\node[mynode] (L4h) at (12.5, 0) {\footnotesize\texttt{h}};
			\draw[myarrow] (L1) to node[midway,above=0pt, left=7pt]{} (L2a);
			\draw[myarrow] (L1) to node[midway,above=0pt, left=-2.5pt]{} (L2b);
			\draw[myarrow] (L2a) to node[midway,above=0pt, left=4pt]{} (L3a);
			\draw[myarrow] (L2a) to node[midway,above=0pt, right=-1pt]{} (L3b);
			\draw[myarrow] (L2b) to node[midway,above=0pt, left=-0.5pt]{} (L3d);
			\draw[myarrow] (L2b) to node[midway,above=0pt, left=-1pt]{} (L3e);
			\draw[myarrow] (L3a) to node[midway,above=0pt, left=0mm]{} (L4a);
			\draw[myarrow] (L3a) to node[midway,above=0pt, right=-2pt]{} (L4b);
			\draw[myarrow] (L3b) to node[midway,above=0pt, left=0mm]{} (L4c);
			\draw[myarrow] (L3b) to node[midway,above=0pt, right=-2pt]{} (L4d);
			\draw[myarrow] (L3d) to node[midway,above=0pt, left=0mm]{} (L4e);
			\draw[myarrow] (L3d) to node[midway,above=0pt, right=-2pt]{} (L4f);
			\draw[myarrow] (L3e) to node[midway,above=0pt, left=0mm]{} (L4g);
			\draw[myarrow] (L3e) to node[midway,above=0pt, right=-2pt]{} (L4h);
		\end{tikzpicture}
		\caption{Pairwise summation SD-tree for the set of 8 summands $\{a,b,c,d,e,f,g,h\}$. This tree has  7 $S$-nodes.}
	\end{subfigure}%

	\caption{The sub-figures show the changing  shape  of pairwise SD-trees as summands are added, and show  how the internal  nodes change (or not) from S to  D and vice versa. New nodes are added in alternating  fashion, so that leaf descendants of all internal nodes are evenly distributed. The $S$-nodes are grey, and children of $S$-nodes have the same number of leaf descendants.}
	\label{fig_pairwise}		
\end{figure}
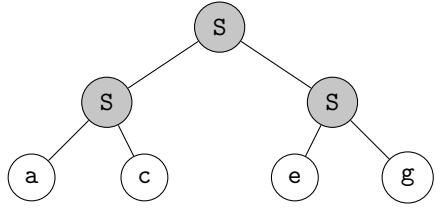
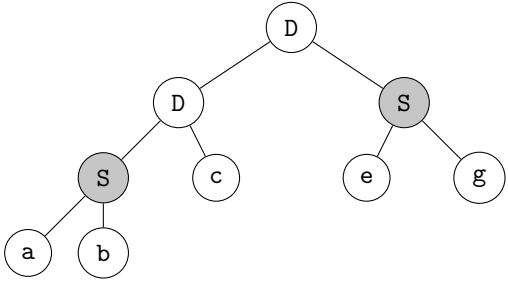
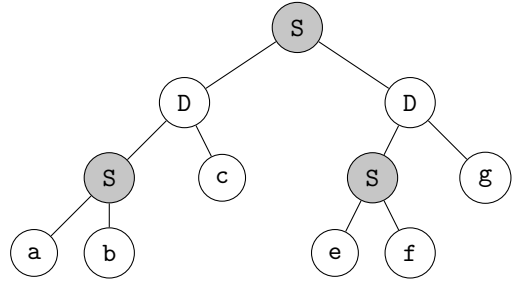
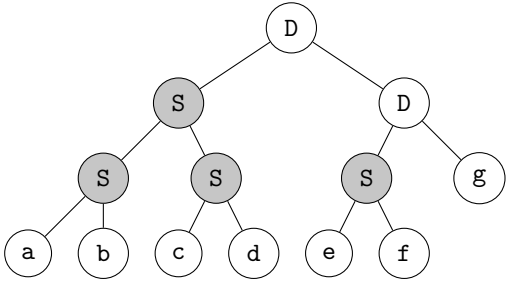
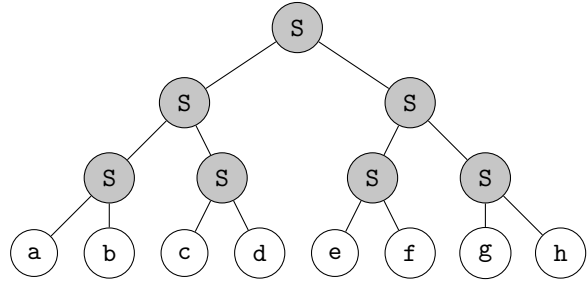

\subsubsection{A formula for pairwise summation using SD-trees.}
We show here a new formula for this sequence, based on the tree representation of a pairwise summation, and the method presented using the SD-tree-based approach from Corollary \ref{equiv_sum}.

Pairwise summation gives rise to a binary tree in which  every node that is not a leaf has exactly two children, and every level is completely filled except for the lowest. This is a \textit{full complete binary tree}, except in the spacing of its lowest level \cite{knuth97art1}. The lowest level is not filled from left to right, but instead the number of descendants of a node differs by at most one from the number of descendants of its siblings, so at every level of the tree, the children of a node are evenly or almost evenly divided between the left and the right branches. Examples are shown in Figure \ref{fig_pairwise}. 
\begin{lemma}\label{pair_iso}
	All pairwise-summation SD-trees on $n$ elements are isomorphic, and therefore have the same number of  $S$-nodes.
\end{lemma}
\begin{proof}
	A pairwise-summation SD-tree is constructed by extending each node by two children. If a node has   $k$  leaf  descendants,  then its children must have $\floor{\frac{k}{2}}$ and $\ceil{\frac{k}{2}}$ leaf  descendants. These  can be transposed, by pairwise  commutativity of addition. That they have the  same  number of $S$-nodes  follows  by Lemma \ref{iso_SD}.
\end{proof}
\begin{lemma}
	Any summation with a summation tree isomorphic to a pairwise-summation SD-tree is itself a pairwise summation.
\end{lemma}
\begin{proof}
	Follows from the definition of pairwise summation.
\end{proof}

\paragraph{Counting $S$-nodes in a pairwise summation.}
We now know that all pairwise summation SD-trees have the same number of $S$ nodes, and all we need do is count them to obtain the exponent in the denominator of the enumeration. Define $\epsilon(n)$ be the number of $S$ nodes in an SD-tree. We show here two different forms for the sequence $\epsilon(n)$. 
\begin{proposition}\label{exp_recursive}
	The number of $S$-nodes in a pairwise summation SD-tree with $n$ leaf nodes is
	\begin{equation}\label{recursive_A268289}
		\varepsilon(n)=
		\begin{cases}
			2\varepsilon(m)+1, & \text{if}\ n=2m; \\
			\varepsilon(m)+\varepsilon(m+1), & \text{if}\ n=2m+1;\\
			0 & \text{if}\ n=1.
		\end{cases}
	\end{equation}
\end{proposition}
\begin{proof}
	All subtrees of a pairwise summation SD-tree are themselves pairwise summation SD-trees. We proceed via induction. 
	
	If $n=2m$, then the pairwise summation divides the $2m$ leaf nodes evenly between the two subtrees, and the root node is an $S$-node. So the number of $S$-nodes in the SD-tree is equal to $1+2\varepsilon(m)$, by the induction hypothesis.
	
	If $n=2m+1$, then the pairwise summation divides the $2m + 1$ leaf nodes into two subtrees having $m$ and $m+1$ nodes, and the root node is not an $S$-node. So the number of $S$-nodes in the SD-tree is equal to $\varepsilon(m)+\varepsilon(m+1)$, by the induction hypothesis.
\end{proof}
\begin{proposition}\label{exp_closed}
	The number of $S$-nodes in a pairwise summation SD-tree with $n$ leaf nodes  is
	\begin{equation}\label{closed_A268289}
		\varepsilon(n)=\sum_{i=0}^{\floor{\log_2(n)}}  \left[ \left(\left(\floor{\frac{n}{2^i}}+1\right)\bmod 2\right)\times 2^i+ (-1)^{\left(\left(\floor{\frac{n}{2^i}}+1\right)\bmod 2\right)}\times (n \bmod  2^i)\right]
	\end{equation}
\end{proposition}
\begin{remark}
	Equation (\ref{closed_A268289}) is simpler than it looks. $\floor{\log_2(n)}$ is one less than the number of bits in $n$. The expression $(\floor{\frac{n}{2^i}}+1)\bmod 2)$ is just the negation of the $i^\text{th}$ bit of $n$ and enables the functions seen in Equation \ref{eqB} and \ref{eqC} from $\{0,1\}$ to $\{1,0\}$ and $\{-1,1\}$, respectively.
\end{remark}
\begin{proof}
	Because every SD-tree is a full complete binary tree, there are $2^i$ nodes at level  $i$ except for the bottom level. At each level, the descendant leaf nodes are almost evenly divided between the nodes: at the $i^\text{th}$ level, each of  the $2^i$ nodes has at least $\floor{\frac{n}{2^i}}$ leaf descendants, and $(n\bmod 2^i)$ of these nodes have an additional leaf descendant. 
	
	Consider the binary representation of the number $n$ where the bits are ordered from least significant to most significant.  
	If the $i^\text{th}$ bit of  $n$ is $0$, then $\floor{\frac{n}{2^i}}$ is even and $2^i-(n\bmod 2^i)$ nodes have an even number of  descendants. If the $i^\text{th}$ bit of  $n$ is $1$, then $\floor{\frac{n}{2^i}}$ is odd, and $(n\bmod 2^i)$ nodes have an even number of  descendants. A node in a pairwise summation tree is an $S$-node if it has an even number of leaf descendants. So the number of $S$-nodes at level $i$ is
	\begin{equation}\label{eqA}
		\begin{cases}
			2^i-(n\bmod 2^i), & \text{ if the}\ i^\text{th} \text{ bit  of $n$ is }  0;\\
			(n\bmod 2^i), & \text{ if the}\ i^\text{th} \text{ bit of $n$ is }  1.\\
		\end{cases}	  
	\end{equation}  
	We observe the following two equations: 
	\begin{equation}\label{eqB}
		\left(\floor{\frac{n}{2^i}}+1\right)\bmod 2=
		\begin{cases}
			$1$,& \text{ if the}\ i^\text{th} \text{ bit of $n$ is }  0;\\ 
			$0$,& \text{ if the}\ i^\text{th} \text{ bit of $n$ is }  1.\\ 
		\end{cases}
	\end{equation}	
	\begin{equation}\label{eqC}
		(-1)^{\left(\left(\floor{\frac{n}{2^i}}+1\right)\bmod 2\right)}=
		\begin{cases}
			$-1$,& \text{ if the}\ i^\text{th} \text{ bit of $n$ is }  0;\\ 
			$ 1$,& \text{ if the}\ i^\text{th} \text{ bit of $n$ is }  1.\\ 
		\end{cases}
	\end{equation}	
	Putting  Equations (\ref{eqA}), (\ref{eqB}),  and (\ref{eqC}) together, we see that the number of $S$-nodes at level $i$ is
	$$
		\left(\left(\floor{\frac{n}{2^i}}+1\right)\bmod 2\right)
		\times 2^i+ (-1)^{\left(\left(\floor{\frac{n}{2^i}}+1\right)\bmod 2\right)}\times (n \bmod  2^i)
	$$
	Summing across the levels proves the proposition.
\end{proof}
\begin{remark}
	Equation (\ref{closed_A268289}) is easily calculated as bitwise operations. Here is one C encoding:
	\begin{verbatim}
	int bits_n = (int) floor(log2(n))+1;
	int exp = 0;
	for (int i=0; i<bits_n; i++) {
	    exp += (!((n>>i)&1))*(1<<i) + (-(!((n>>i)&1)<<1)+1)*(n-((n>>i)<<i));
	}
	\end{verbatim}
\end{remark}
\begin{remark}
	The exponent $\varepsilon(n)$ in Equations (\ref{recursive_A268289}) and (\ref{closed_A268289}) is an offset variant of the sequence OEIS A268289 \cite{OEISA268289}, the cumulative deficient binary digit sum. In the context of $S$-nodes, $\varepsilon(n)$ is undefined at $n=0$ and starts with $\varepsilon(1)=0$. 
	The SD-tree equation in Proposition \ref{exp_recursive} is the same formula referenced by Sloane for OEIS A268289 \cite{OEISA268289}. The closed form shown in Equation (\ref{closed_A268289}) is new for that sequence. 
\end{remark}
\begin{remark}
    Hwang provides an extensive analysis of solutions to divide-and-conquer recurrences in \cite{hwang17}. Equation (\ref{closed_A268289}) could also be obtained from Equation (\ref{recursive_A268289}) using the methods in that paper. Instead, we prove Equation (\ref{closed_A268289}) directly by analyzing $S$-nodes in Proposition \ref{exp_closed}.
\end{remark}

\paragraph{Counting equivalent pairwise summations.}
The exponent in the  denominator  of the enumeration of the equivalence class of pairwise  summations has been derived, and the enumeration itself immediately follows. 
\begin{corollary}\label{pairrecursivenew}
	The number of computationally inequivalent pairwise summations on $n$ variables is 
	\begin{equation}\label{recursive2_A096351}
    	\sigma(n)=\frac{n!}{2^{\varepsilon(n)}} \text{, where }
    	\varepsilon(n)=
    	\begin{cases}
        	2\varepsilon(m)+1, & \text{if}\ n=2m; \\
        	\varepsilon(m)+\varepsilon(m+1), & \text{if}\ n=2m+1.\\
        	0, & \text{if}\ n=1.\\
    	\end{cases}
	\end{equation}
\end{corollary}
\begin{proof}
	Follows directly  from  Corollary \ref{equiv_sum} and Proposition \ref{exp_recursive}.
\end{proof}
\begin{corollary}\label{pairclosednew}
	The number of computationally inequivalent pairwise summations on $n$ variables is 
	\begin{equation}\label{closed_A096351}
	\sigma(n)=\frac{n!}{2^{\varepsilon(n)}} \text{, where }
    	\varepsilon(n)=\sum_{i=0}^{\floor{\log_2(n)}}  \left(\left(\floor{\frac{n}{2^i}}+1\right)\bmod 2\right)\cdot 2^i+ (-1)^{\left(\left(\floor{\frac{n}{2^i}}+1\right)\bmod 2\right)}\cdot (n \bmod  2^i)
	\end{equation}
\end{corollary}
\begin{proof}
	Follows directly from Corollary \ref{equiv_sum} and Proposition \ref{exp_closed}.
\end{proof}
\begin{remark}
	Equation (\ref{recursive2_A096351})  is a new recursive formula for OEIS A096351  \cite{OEISA096351} and Equation (\ref{closed_A096351}) is a new closed formula for the same.
\end{remark}

\subsubsection{Pairwise summations, Karatsuba recursion and the Tagaki function.} There are interesting connections between $S$-nodes in a pairwise summation, the Tagaki function   \cite{takagi01} and Karatsuba's classical addition algorithm  \cite{karatsuba62}. OEIS A268289, i.e., the number of $S$-nodes in a pairwise summation, is related to the Tagaki function; there are a number of identities linking the two \cite{ baruchel19_2, lagarias11}.
Baruchel discusses the recursion tree of Karatsuba's multiplication algorithm in terms of two types of nodes in the tree, direct and indirect. He conjectured in \cite{baruchel19} and proved in \cite{baruchel19_2} that the sequence $\varepsilon(n+1)$ in Equation (\ref{baruchel_set}), the set of elements in the indirect path that  contribute to a term of degree $n$ in the Karatsuba addition, is A$268289_{n+1}$.  

We provide a direct proof of Baruchel's identity on $\varepsilon(n)$ (A268289) here, but in terms of the number of $S$-nodes of pairwise summations (or full complete binary trees with the lowest levels spaced as in a pairwise summation), rather than the Karatsuba recursion tree.

To prove this, we consider an $n$-leafed pairwise-summation tree $T$, and label the $(n-1)$ internal nodes of $T$ breadth-first with the integers from $1$ to $(n-1)$. The tree $T$ has two subtrees: the left $T_L$ and the right $T_R$. The internal nodes of these could be labeled in the same manner as those of $T$.
 Without loss of generality, if $n=2m+1$, $T_L$ has the greater number ($m+1$) of leaf nodes and $T_R$ has $m$ leaf nodes. If $n=2m$, then both $T_L$ and $T_R$ have $m$ leaf nodes. This  ordering of subtrees is key to the mapping from $T_L$ and $T_R$ to $T$ that we define for the proof.

 We  define $f$, mapping the internal nodes of $T_L$ and $T_R$ into the internal nodes of $T$:
\begin{equation*}
    f(x) =
    \begin{cases}
        2x,  & \text{ if } x \in T_L; \\
        2x+1,  & \text{ if } x \in T_R. \\
    \end{cases}
\end{equation*}
The function $f$ is a bijection between the internal nodes of $T_L \cup T_R$ and the non-root internal nodes of $T$. This mapping is illustrated in Figure \ref{fig_karatsuba}.

\begin{figure}[h!tb]
	
	\tikzstyle{mynode}=[circle, draw]
	\tikzstyle{rectanglenode}=[rectangle, draw]
	\tikzstyle{Snode}=[fill=gray!50]
	\tikzstyle{sqSnode}=[fill=gray!20]
	\tikzstyle{leaf}=[fill=gray!45]
	\tikzstyle{indirect}=[pattern=north east lines, pattern color=gray!80]
	\tikzstyle{myarrow}=[]
	\centering
	\begin{subfigure}[t]{.45\textwidth}
		\begin{tikzpicture}	
			\centering
			\node[mynode,Snode] (L2a) at (-1.5, 2) {\footnotesize\texttt{1}};
			\node[mynode] (L2b) at (1.5, 2) {\footnotesize\texttt{1}};
			\node[mynode,Snode] (L3a) at (-2.5, 1) {\footnotesize\texttt{2}};	
			\node[mynode,Snode] (L3b) at (-1, 1) {\footnotesize\texttt{3}};
			\node[mynode,sqSnode, rectangle, minimum height = .6cm, minimum width = .6cm] (L3d) at (1, 1) {\footnotesize\texttt{2}};
			\node[mynode] (L3e) at (2.5, 1) {\footnotesize\texttt{}};
			\node[mynode] (L4a) at (-3.5, 0) {\footnotesize\texttt{}};
			\node[mynode] (L4b) at (-2.5, 0) {\footnotesize\texttt{}};
			\node[mynode] (L4c) at (-1.5, 0) {\footnotesize\texttt{}};
			\node[mynode] (L4d) at (-0.5, 0) {\footnotesize\texttt{}};
			\node[mynode] (L4e) at (0.5, 0) {\footnotesize\texttt{}};
			\node[mynode] (L4f) at (1.5, 0) {\footnotesize\texttt{}};
			\draw[myarrow] (L2a) to node[midway,above=0pt, left=4pt]{} (L3a);
			\draw[myarrow] (L2a) to node[midway,above=0pt, right=-1pt]{} (L3b);
			\draw[myarrow] (L2b) to node[midway,above=0pt, left=-0.5pt]{} (L3d);
			\draw[myarrow] (L2b) to node[midway,above=0pt, left=-1pt]{} (L3e);
			\draw[myarrow] (L3a) to node[midway,above=0pt, left=0mm]{} (L4a);
			\draw[myarrow] (L3a) to node[midway,above=0pt, right=-2pt]{} (L4b);
			\draw[myarrow] (L3b) to node[midway,above=0pt, left=0mm]{} (L4c);
			\draw[myarrow] (L3b) to node[midway,above=0pt, right=-2pt]{} (L4d);
			\draw[myarrow] (L3d) to node[midway,above=0pt, left=0mm]{} (L4e);
			\draw[myarrow] (L3d) to node[midway,above=0pt, right=-2pt]{} (L4f);
		\end{tikzpicture}
		\caption{Two pairwise summation SD-trees, the left, $T_L$, with 4 leaves and the right, $T_R$, with 3. Internal nodes are numbered breadth-first, and the $S$-nodes are grey.}
	\end{subfigure}\hfill%
	\begin{subfigure}[t]{.5\textwidth}
		\begin{tikzpicture}	
			\centering
			\node[mynode] (L1) at (9, 3) {\footnotesize\texttt{1}};
			\node[mynode,Snode] (L2a) at (7.5, 2) {\footnotesize\texttt{2}};
			\node[mynode] (L2b) at (10.5, 2) {\footnotesize\texttt{3}};
			\node[mynode,Snode] (L3a) at (6.5, 1) {\footnotesize\texttt{4}};
			\node[mynode,sqSnode, rectangle, minimum height = .6cm, minimum width = .6cm] (L3b) at (8, 1) {\footnotesize\texttt{5}};
			\node[mynode,Snode] (L3d) at (10, 1) {\footnotesize\texttt{6}};
			\node[mynode] (L3e) at (11.5, 1) {\footnotesize\texttt{}};
			\node[mynode] (L4a) at (5.5, 0) {\footnotesize\texttt{}};
			\node[mynode] (L4b) at (6.5, 0) {\footnotesize\texttt{}};
			\node[mynode] (L4c) at (7.5, 0) {\footnotesize\texttt{}};
			\node[mynode] (L4d) at (8.5, 0) {\footnotesize\texttt{}};
			\node[mynode] (L4e) at (9.5, 0) {\footnotesize\texttt{}};
			\node[mynode] (L4f) at (10.5, 0) {\footnotesize\texttt{}};
			\draw[myarrow] (L1) to node[midway,above=0pt, left=7pt]{} (L2a);
			\draw[myarrow] (L1) to node[midway,above=0pt, left=-2.5pt]{} (L2b);
			\draw[myarrow] (L2a) to node[midway,above=0pt, left=4pt]{} (L3a);
			\draw[myarrow] (L2a) to node[midway,above=0pt, right=-1pt]{} (L3b);
			\draw[myarrow] (L2b) to node[midway,above=0pt, left=-0.5pt]{} (L3d);
			\draw[myarrow] (L2b) to node[midway,above=0pt, left=-1pt]{} (L3e);
			\draw[myarrow] (L3a) to node[midway,above=0pt, left=0mm]{} (L4a);
			\draw[myarrow] (L3a) to node[midway,above=0pt, right=-2pt]{} (L4b);
			\draw[myarrow] (L3b) to node[midway,above=0pt, left=0mm]{} (L4c);
			\draw[myarrow] (L3b) to node[midway,above=0pt, right=-2pt]{} (L4d);
			\draw[myarrow] (L3d) to node[midway,above=0pt, left=0mm]{} (L4e);
			\draw[myarrow] (L3d) to node[midway,above=0pt, right=-2pt]{} (L4f);
		\end{tikzpicture}
		\caption{A pairwise summation SD-tree $T$ with $7$ leaves. The left subtree is  isomorphic to  $T_L$, and  the right  subtree is isomorphic to  $T_R$. Internal nodes are numbered breadth-first, and $S$-nodes are grey. }
	\end{subfigure}%
	\caption{The mapping  $f$  is illustrated here. Internal nodes $x$ from $T_L$ are sent to $(2x)$, so $1$, $2$, and $3$ from $T_L$ (circles) are mapped to internal  nodes $2$, $4$, and $6$ of tree $T$. Internal nodes $x$ from $T_R$ are sent to $(2x+1)$, so internal node $2$ of  $T_R$ (square) is mapped  to internal  node $5$ of  $T$, and internal node $1$ of $T_R$ is mapped to node $3$ of $T$.}
	\label{fig_karatsuba}		
\end{figure}
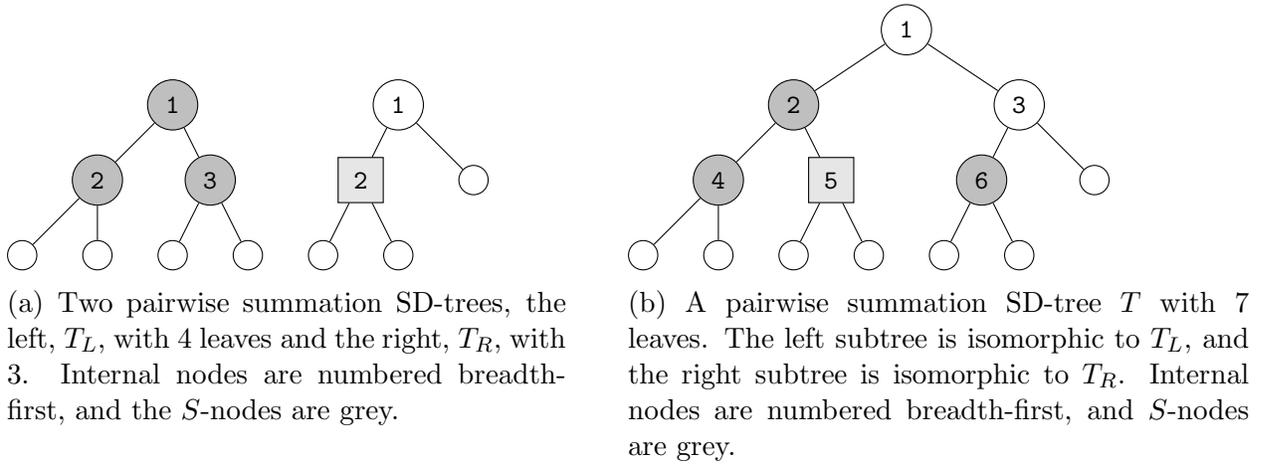

\begin{proposition}\label{baruchel_conjecture}
    Let $\varepsilon(n)$ be the number of $S$-nodes in a pairwise summation SD-tree. Then
	\begin{equation}\label{baruchel_set}
    	\varepsilon(n)=| \{k \mid 1 \le k < n, (n-k-1) \bmod 2^{\floor{\log_2 k}+1} < 2^{\floor{\log_2 k}} \}|
	\end{equation}
\end{proposition}
\begin{proof}
	It suffices  to  show that the $S$-nodes of $T$ are  exactly those  whose labels are in the set described by Equation (\ref{baruchel_set}). 
	The proof proceeds by induction on tree height. We consider the two cases: $n=2m+1$  ($n$ odd) and $n=2m$ ($n$ even). 
	
	First,  the tree root satisfies condition \ref{baruchel_set} if and only if $n$ is even: the  root is  labeled  $1$, so $(n-1-1) \bmod 2^{\floor{\log_2 1}+1} < 2^{\floor{\log_2 1}}$ if and only if $(n-1-1)=0\bmod 2$ if and only if $n$ is even. 

	\emph{Case  1:} $n$ odd, $n=2m+1$. 
	By the inductive hypothesis, $x$ is an $S$-node of $T_L$ if and only if   
	\begin{equation*}
		((m+1)-x-1) \bmod 2^{\floor{\log_2(x)}+1}  < 2^{\floor{\log_2(x)}} \text{, which is true if and only if } 
	\end{equation*}
	\begin{equation*}
		(2(m+1)-2x-2) \bmod 2^{\floor{\log_2(x)}+2}  < 2^{\floor{\log_2(x)}+1} \text{, which is the same  as}
	\end{equation*}
	\begin{equation*}
		(n-2x-1) \bmod 2^{\floor{\log_2(2x)}+1}  < 2^{\floor{\log_2(2x)}} \text{, which is true if and only if}
	\end{equation*}
	\begin{equation}\label{eqEven}
		(n-k-1) \bmod 2^{\floor{\log_2(k)}+1}  < 2^{\floor{\log_2(k)}} \text{, where  $k$ is  an even node of $T$.}
	\end{equation}
	Again, by the induction hypothesis, $x$ is an $S$-node of $T_R$  if and only if
	\begin{equation*}
		(m-x-1) \bmod 2^{\floor{\log_2(x)}+1}  < 2^{\floor{\log_2(x)}} \text{, which is true if and only if}
	\end{equation*}
	\begin{equation*}
		(2m-2x-2) \bmod 2^{\floor{\log_2(2x+1)}+1}  < 2^{\floor{\log_2(x+1)}} \text{, which is the same  as}
	\end{equation*}
	\begin{equation*}
		((n-1)-(2x+1)-1) \bmod 2^{\floor{\log_2(2x+1)}+1}  < 2^{\floor{\log_2(x+1)}} \text{, which is true if and only if}
	\end{equation*}
	\begin{equation}\label{eqOdd}
		(n-k-1) \bmod 2^{ \floor{\log_2(k)} +1}  < 2^{\floor{\log_2(k)}} \text{, where $k$ is an odd node of $T$, and } k\ge 3.
	\end{equation}
	Statement (\ref{eqEven}) shows there is a bijection between even nodes of $T$ that satisfy the condition in the proposition and the $S$-nodes of $T_L$, and statement (\ref{eqOdd}) shows a bijection between odd non-root nodes of $T$ that satisfy the condition and the $S$-nodes of $T_R$.	Since $n$ is odd, the root node does not satisfy the condition. So the number of nodes in $T$ that satisfy the condition in Equation (\ref{baruchel_set}) is $\varepsilon(m+1) + \varepsilon(m)$,  and this is   $\varepsilon(n)$, by Proposition \ref{exp_recursive}.

	\emph{Case  2:} $n$ even, $n=2m$. 	
	The proof in this case proceeds in the same manner as the above. By the induction hypothesis, $x$ is an $S$-node of $T_L$ if and only if   
    \begin{equation*}
    	(m-x-1) \bmod 2^{\floor{\log_2(x)}+1}  < 2^{\floor{\log_2(x)}} \text{, which is true if and only if }
    \end{equation*}
    \begin{equation*}
    	(2m-2x-2) \bmod 2^{\floor{\log_2(x)}+2}  < 2^{\floor{\log_2(x)}+1} \text{, which is the same  as}
    \end{equation*}
    \begin{equation*}
    	((n-1)-2x-1) \bmod 2^{\floor{\log_2(2x+1)}+1}  < 2^{\floor{\log_2(2x+1})} \text{, which is true if and only if}
    \end{equation*}
    \begin{equation}\label{eqEven2}
        (n-k-1) \bmod 2^{\floor{\log_2(k)}+1}  < 2^{\floor{\log_2(k)}} \text{, where  $k$ is  an even node of $T$.} 
    \end{equation}
    Again, by the induction hypothesis, $x$ is an $S$-node of $T_R$  if and only if
    \begin{equation*}
    	(m-x-1) \bmod 2^{\floor{\log_2(x)}+1}  < 2^{\floor{\log_2(x)}} \text{, which is true if and only if}
    \end{equation*}
    \begin{equation*}
    	(2m-2x-2) \bmod 2^{\floor{\log_2(x)}+2}  < 2^{\floor{\log_2(x)}+1} \text{, which is true if and only if}
    \end{equation*}
    \begin{equation*}
    	(n-(2x+1)-1) \bmod 2^{\floor{\log_2(2x+1)}+1}  < 2^{\floor{\log_2(2x+1)}} \text{, which is true if and only if}
    \end{equation*}
    \begin{equation}\label{eqOdd2}
    	(n-k-1) \bmod 2^{\floor{\log_2(k)}+1}  < 2^{\floor{\log_2(k)}} \text{, where $k$ is an odd node of $T$, and } k\ge 3.
    \end{equation}
	 Statement (\ref{eqEven2}) shows there is a bijection between even nodes of $T$ that satisfy the condition in the proposition and the $S$-nodes of $T_L$, and (\ref{eqOdd2}) shows a bijection between odd non-root nodes of $T$ that satisfy the condition and the $S$-nodes of $T_R$.	The root node does satisfy the condition when $n$ is even. So the number of nodes in $T$ that satisfy the condition in Equation (\ref{baruchel_set}) is $2\varepsilon(m) + 1$,  and this is   $\varepsilon(n)$, by Proposition \ref{exp_recursive}.
\end{proof}

\subsubsection{A remark on orderings for pairwise summations vs. tournaments  }
The parenthetic form of grouping for a pairwise summation is the same as the form for a tournament. However, the optimal ordering for a summation may well be different than that of a tournament. This contrast illustrates the problem-specific nature of selecting the preferred or best ordering, even after a grouping method has been selected.

As a heuristic,  the most accurate pairwise summation groups pairs of similar magnitude at each level of the summation tree. This tends to minimize rounding error that results from the use of finite precision on digital computers \cite{job20}.

The best tournament might also match teams of equal strength at each level of the tournament, leading to interesting games, at least in early games. This is equivalent to the best ordering for a pairwise summation. On the other hand, one might prefer better teams to be matched with worse teams early in the competition, in order to allow the best teams to persevere into the final rounds of the tournament. This would be a very poor choice for a pairwise summation, and is more likely to lead to rounding errors and inaccuracy. 
The choice of ordering here then would depend on the goals of the tournament.

\section{Bounds and enumerations}

\subsection{Bounds  on inequivalent summations on a parenthetic form}
\begin{proposition}\label{upper_bound}
	The upper bound for the number of computationally inequivalent summations on $n$ summands on a class of isomorphic SD-trees is $\frac{n!}{2}$. This is the number of ladder summations on $n$ summands. 
\end{proposition}
\begin{proof}
	This can be seen by considering that all nodes in a summation tree must have zero or two children, so leaf nodes at the lowest level must come in pairs, and any tree must have at least one of these pairs. This means an upper bound is less than or equal to $\frac{n!}{2}$, by Corollary \ref{equiv_sum}. But this bound is met by the class of ladder summations on $n$ terms. 
\end{proof}
\begin{remark}  
	This upper bound is OEIS A001710 \cite{OEISA001710}, the number of even permutations on $n$ letters. See Proposition \ref{numladder}.
\end{remark}
\begin{proposition}\label{lower_bound}
	 The lower bound for the number of summations on $n$ summands on a class of isomorphic SD-trees is
	 $\frac{n!}{2^{\beta(n)}}$, where $2^{\beta(n)}$ is the highest  power of  $2$  that  divides $n!$. 
\end{proposition}
\begin{proof}
	Follows from Corollary \ref{equiv_sum}.
\end{proof}
\begin{remark}  
The sequence $\beta(n)$ is OEIS A011371 \cite{OEISA011371}, the highest power of $2$ dividing $n!$. The sequence $\frac{n!}{2^{\beta(n)}}$ is OEIS A049606 \cite{OEISA049606}, the largest odd divisor of $n!$.
\end{remark}
\begin{lemma}\label{pow2_2k}
	The largest power of $2$ that divides $2^k!$ is  $2^{(2^k-1)}$.
\end{lemma}
\begin{proof}
	Let $2^{\beta(2^k)}$ be the largest power of $2$ that divides $2^k!$. For  $1\le i \le k$ there are $\frac{2^k}{2^i}=2^{k-i}$ numbers in $\{1,2,\dots,2^k\}$ divisible by $2^i$. Each of these adds one to  $\beta(2^k)$. So $\beta(2^k)=\sum\limits_{i=1}^{k}2^{k-i}=\sum\limits_{i=0}^{k-1}2^{i}=2^k-1$.
\end{proof}
\begin{lemma}\label{pow2_m}
	Let $n=2^k+r$, where $0  < r <2^k$. 
	The largest power of $2$ that divides $\frac{n!}{(2^k)!}=n(n-1)\cdots (2^k+1)$ is the same as the largest power of $2$ that divides $(n-2^k)!=r!$.
\end{lemma}
\begin{proof}
	$2^k$ does  not divide any of $\{2^k+1, 2^k+2, \dots ,2^k+r=n\}$, so largest power of of $2$ that divides $\frac{n!}{(2^k)!}=n(n-1)\cdots (2^k+1)$ is the same as the largest power of $2$ that divides $(n-2^k)!$.
\end{proof}	
\begin{lemma}\label{pow2_2km}
	Let $n=2^k+r$, where $0  < r <2^k$. The largest power of $2$ that divides $n!$ is the largest power of $2$ that divides $2^k!$ plus the largest power of $2$ that divides $(n-2^k)!=r!$
\end{lemma}
\begin{proof}
	Follows from Lemmas \ref{pow2_2k}  and \ref{pow2_m}.
\end{proof}	
\begin{proposition}\label{lowerboundmet}
	Let $n=2^k+r$ , where $0 \le r < 2^k$. The bound in Proposition \ref{lower_bound} on summations on $n$ summands is met by a parenthetic form $\mu(n)$ that has the recursive  form 
		$$\mu(n)=(\mu(2^k)+\mu(r))$$
where $\mu(2^i)$ is the pairwise summation on $2^i$ elements.
\end{proposition}
\begin{proof}[Case 1, $r=0$] 
	A pairwise summation on $2^k$ elements is represented by a perfect binary tree, one in which all interior nodes have two children, and all leaf nodes are at the same level $k$ \cite{zou19}. A perfect binary tree has $2^{k}$ leaf nodes and $2^k - 1$  interior nodes. Each of the interior nodes is an $S$-node. So by Corollary \ref{equiv_sum} and Lemma \ref{pow2_2k}, the bound is  met.
\end{proof}
\begin{proof}[Case 2, $r>0$]
	$2^k$ is not equal to $n-2^k$, so the top node is not an $S$-node  and the problem reduces to finding the largest powers of 2 that divide $2^k$ and $r$. The conclusion follows by Lemma \ref{pow2_2km}.
	\end{proof}
\begin{remark}
    The parenthetic form shown in Proposition \ref{lowerboundmet} is not necessarily the only one that meets the bound. For example, Table \ref{paren_nk} shows $3$ non-equivalent parenthetic forms meeting the bound  $\beta(n)$  from Proposition \ref{upper_bound} for each of $n=7, 11, 13, 14$, and $15$ non-equivalent parenthetic forms meeting the bound for $n=15$.
\end{remark}

\subsection{Enumeration of parenthetic forms}
\begin{proposition}\label{halfcat}
	The number of non-isomorphic summation trees (or parenthetic forms) with $n$ unlabeled leaf nodes (or summands) is 
	\begin{equation}
	\alpha(n)=\sum_{i=1}^{\floor{\frac{n}{2}}}\alpha(i)\alpha(n-i)
	\end{equation}
\end{proposition}
\begin{proof}
    Partition $n$ as $i$ and $n-i$. A summation tree is formed by attaching an $i$-leaf tree as the left subtree and an $(n-i)$-leaf tree as the right. Because of commutativity, it suffices to consider only the case where $i\le n-i$. Summing across such $i$ gives the proposition.
\end{proof}
\begin{remark}
	This sequence  is  OEIS A000992, the $n^\text{th}$ half-Catalan number. An equivalent observation is made by Callan \cite{OEISA000992}. In Table \ref{paren_nk}, this sequence is  in the rightmost column.
\end{remark}
\begin{proposition}\label{propkSnodes}
	The number of non-isomorphic summation trees  (or parenthetic forms) with $n$ unlabeled leaf nodes (or summands) and $s$ $S$-nodes is  
	\begin{equation}
		\tau(n,s) = 
		\begin{cases}
			\sum\limits_{j=1}^{\floor{\frac{n-1}{2}}}\sum\limits_{i=0}^s \tau(j,i)\tau(n-j,s-i), & \text{if $n$ is odd}; \\
			\sum\limits_{j=1}^{\floor{\frac{n-1}{2}}}\sum\limits_{i=0}^s \tau(j,i)\tau(n-j,s-i) + \sum\limits_{i=0}^{s-1} \tau(\frac{n}{2},i)\tau(\frac{n}{2},s-1-i), & \text{if $n$ is even}. \\
		\end{cases}
	\end{equation}
	\begin{equation*}
		\text{where } \tau(n,s)=
		\begin{cases}
			1, & \text{if } n=0 \text{ and } s=0 \text{ (the unique empty tree)};  \\
			1, & \text{if } n=1 \text{ and } s=0 \text{ (the unique single-node tree)}.\\
		\end{cases}
	\end{equation*}
\end{proposition}
\begin{proof}
	At every stage, the first term $n$ in the calculation is reduced by half, so this recursion terminates, according to the rules for termination of $\tau(n,s)$ above.
	
	If $n$ is odd, the children of the top node cannot both have the same number of leaf descendants, so the top node cannot be an $S$-node. The formula follows by counting the sub-trees on the 2-partitions of $n$, where the left sub-tree has $i$ $S$-nodes  and the right has $s-i$ $S$-nodes, for  $0 \le i \le  s$, and the two counts are multiplied together.
	
	If $n$ is even,  the formula is as above, but we must also include the case where the top node is an S-node. In that case, both children have $\frac{n}{2}$ leaf nodes, and the summation is on sub-trees where there are  $\frac{n}{2}$ leaf nodes in each sub-tree. In this case, the root is one of the $S$-nodes, so the left sub-tree has $i$ $S$-nodes  and the right has $s-1-i$ $S$-nodes, for  $0 \le i \le  s-1$. 
\end{proof}
\begin{remark}
	The number of parenthetic forms increases much more  slowly than the total number of summations. Table \ref{paren_nk} shows the enumeration of parenthetic forms on $n$ summands and $s$ $S$-nodes for $n$ up to 16 summands. 
\end{remark}
\begin{remark}
	The formula in Proposition  \ref{propkSnodes} is recursive and the proof is constructive, so one can more or less laboriously calculate all the parenthetic forms on $N$ summands.
\end{remark}

	\begin{table}[h!tb]
	\centering
	\caption{The number $\tau(n,s)$ of summation trees with $n$ leaf nodes and $s$ $S$-nodes, calculated up to $n=15$. The  rows  represent $n$, the number of leaf nodes, and the columns represent $s$, the number of $S$-nodes.  $\alpha(n)$ is the sum of all the $\tau(n,s)$ in the $n^\text{th}$ row.}
		\begin{tabular}{r|rrrrrrrrrrrrrr|r}
			\backslashbox{\textbf{n}}{\textbf{s}} & \textbf{1} & \textbf{2}  & \textbf{3}   & \textbf{4}    & \textbf{5}    & \textbf{6}    & \textbf{7}    & \textbf{8}    & \textbf{9}   & \textbf{10}  & \textbf{11} & \textbf{12} & \textbf{13} & \textbf{14}   & $\boldsymbol{\alpha(n)}$\\
			\hline
			\textbf{1}   &    &    &     &      &      &      &      &      &     &     &    &    &    &    &    \textbf{0}  \\
			\textbf{2}   & 1 &   &     &      &      &      &      &      &     &     &    &    &    &    &    \textbf{1}   \\
			\textbf{3}   & 1 & 0  &    &      &      &      &      &      &     &     &    &    &    &    &    \textbf{1}   \\
			\textbf{4}   & 1 & 0  & 1   &     &      &      &      &      &     &     &    &    &    &    &    \textbf{2}   \\
			\textbf{5}   & 1 & 1  & 1   & 0    &      &      &      &      &     &     &    &    &    &    &    \textbf{3}  \\
			\textbf{6}   & 1 & 2  & 2   & 1    & 0    &     &      &      &     &     &    &    &    &    &     \textbf{6}  \\
			\textbf{7}   & 1 & 4  & 3   & 3    & 0    & 0    &       &      &     &     &    &    &    &    &  \textbf{11}    \\
			\textbf{8}   & 1 & 6  & 7   & 6    & 3    & 0    & 1    &      &     &     &    &    &    &    &     \textbf{24}  \\
			\textbf{9}   & 1 & 9  & 14  & 13   & 8    & 1    & 1    & 0    &     &     &    &    &    &    &     \textbf{47}  \\
			\textbf{10}  & 1 & 12 & 27  & 28   & 23   & 8    & 3    & 1    & 0   &      &    &    &    &    &    \textbf{103}   \\
			\textbf{11}  & 1 & 16 & 49  & 58   & 54   & 25   & 8    & 3    & 0   & 0   &     &    &    &    &   \textbf{214}    \\
			\textbf{12}  & 1 & 20 & 82  & 119  & 125  & 82   & 34   & 15   & 2   & 1   & 0  &    &    &    &   \textbf{481}   \\
			\textbf{13}  & 1 & 25 & 132 & 237  & 270  & 213  & 99   & 42   & 8   & 3   & 0  & 0  &    &    &  \textbf{1030}   \\
			\textbf{14}  & 1 & 30 & 199 & 449  & 578  & 542  & 322  & 151  & 51  & 11  & 3  & 0  & 0  &    &    \textbf{2337}   \\
			\textbf{15}  & 1 & 36 & 294 & 821  & 1190 & 1255 & 867  & 440  & 173 & 39  & 15 & 0  & 0  & 0  &   \textbf{5131}  \\ 
		\end{tabular}
		\label{paren_nk}
	\end{table}

\begin{proposition}\label{prop2Snodes}
	The number of parenthetic forms on $n \ge 1$ with exactly $2$ $S$-nodes is
	\begin{equation}
	\tau(n,2) =
	\begin{cases}
	(m-1)^2, &   \text{if } n=2m+1; \\
	(m-1)(m-2), &   \text{if } n=2m. \\
	\end{cases}
	\end{equation}
\end{proposition}
\begin{proof}
	We proceed by induction, and observe that the proposition is  true by inspection for $n=1,2,3,4$. We also observe that  no tree having exactly 2 $S$-nodes can have an $S$-node root. This means that either 
	(a) one of the subtrees below the root has two $S$-nodes and the other has none, or else 
	(b) both subtrees have exactly one $S$-node.  
	\newline
	 In case (a), the subtree with no $S$-nodes consists of only one leaf node. The other subtree has $n-1$ leaf nodes and $2$ $S$-nodes. So the number of trees  with such subtrees is $\tau(n-1,2)$. 
	 \newline
	 In case (b), each subtree has one $S$-node, so by Corollary \ref{laddercor}, it is the unique ladder tree of its size. Each subtree with $\ell$ leaf nodes has a sibling with $n-\ell$ leaf nodes, and there is one ladder tree at each of these sizes, so the number of such double ladder trees is $\floor{\frac{n}{2}}-1$ if $n$ is odd, and ${\frac{n}{2}}-2$ if $n$ is even (excluding when the siblings have an equal number of leaf nodes).
	 \newline
	 Putting (a) and (b) together, $\tau(n,2)=\tau(n-1,2)+\floor{\frac{n}{2}}-1$. Using the induction hypothesis,
	 \begin{equation*}
    	 \tau(n,2) =
    	 \begin{cases}
        	 (m-1)(m-2)+(m-1) = (m-1)^2, &  \text{if } n=2m+1; \\
        	 (m-2)^2+(m-2)= (m-1)(m-2), &   \text{if } n=2m. \\
    	 \end{cases}
	 \end{equation*}
\end{proof}
\begin{remark}
	$\tau(n,1)$  counts the (unique) ladder summations, as discussed in Proposition \ref{numladder}. $\tau(n,2)$ is sequence OEIS A002620 \cite{OEISA002620}, the quarter-squares, with an offset. The bold-face column $\boldsymbol{\alpha(n)}$ in Table \ref{paren_nk} is the half-Catalan sequence OEIS A000992 \cite{OEISA000992} seen in Proposition \ref{halfcat}. None of the other sequences in Table \ref{paren_nk} seem to be listed in the OEIS.
\end{remark}

\section{Conclusion}
We have classified a number of summation types arising in real calculations. Table \ref{summary_table}  in the Appendix presents a summary of the main combinatorial results derived or referenced in this paper, with other interpretations and OEIS numbers where applicable.

\section{Acknowledgments}

We thank Terry Grov\'e for asking the questions that inspired this work, and for many additional conversations. We also thank Andy DuBois, Shane Fogerty, Brett Neuman, Chris Mauney, and Bob Robey for many interesting and helpful discussions on the subject.

This work was performed at Los Alamos National Laboratory, managed by Triad National Security, LLC, under U.S. Government contract 89233218CNA000001, and at LANL's Ultrascale Systems Research Center (USRC) at the New Mexico Consortium, supported by the U.S. Department of Energy contract DE-FC02-06ER25750.

This publication is unclassified, and has been assigned LANL identifier LA-UR-20-23426.

In memory of John Conway and Vera Pless.

\bibliographystyle{ieeetr}

\bigskip
\hrule
\bigskip

\noindent 2010 {\it Mathematics Subject Classification}:
Primary 05A10, Secondary 05C30, 65G50.

\noindent \emph{Keywords: }
summations, numerical error, binary trees

\bigskip
\hrule
\bigskip

\noindent (Concerned with sequences
\seqnum{A000992},
\seqnum{A001147},
\seqnum{A001710},
\seqnum{A002620},
\seqnum{A011371},
\seqnum{A060818},\\
\seqnum{A049606},
\seqnum{A096351} and
\seqnum{A268289}.)

\bigskip
\hrule
\bigskip
\newpage

\section{Appendix: Summary of results}

\begin{table}[h!tb]
	\small 
	\centering
	\caption{This table is a summary of the main combinatorial results derived or referenced in this paper, with references, alternative interpretations and OEIS numbers where applicable.}
	 \label{summary_table}

	 \newcolumntype{C}{c}
	 \newcolumntype{T}{>{\raggedright\hspace{0pt}}p{0.45\linewidth}|>{\raggedright\hspace{0pt}}p{0.12\linewidth}|>{\raggedright\hspace{0pt}}p{0.21\linewidth}|p{0.09\linewidth}}
	 
    \begin{tabular} {T} 
	\textbf{Summation interpretation} & \textbf{Refs.} & \textbf{Other meaning} & \textbf{OEIS} \\
	\hline
	\end{tabular}

    \begin{tabular} {C}
    	\\[-.75em]
    	\textit{\textbf{Inequivalent parenthetic forms  (i.e., non-isomorphic summation trees)}} \\
    	\\[-.75em]
	\end{tabular}

    \begin{tabular} {T} 
		\hline
    	No. of non-isomorphic summation trees with $n$ unlabeled leaf nodes & Prop. \ref{halfcat} & Half-Catalan numbers & \seqnum{A000992}\\
    	\hline
    	No. of non-isomorphic summation trees with $n$ unlabeled leaf nodes and 1 $S$-node (ladder or serial summations) & Cor. \ref{laddercor} & The constant 1 &  \\
		\hline
    	No. of non-isomorphic summation trees with $n$ unlabeled leaf nodes and 2 $S$-nodes& Prop. \ref{prop2Snodes} & Quarter-squares & \seqnum{A002620} \\
		\hline
    	No. of non-isomorphic summation trees with $n$ unlabeled leaf nodes and s $S$-nodes & Prop. \ref{propkSnodes} &  &  \\
		\hline
    \end{tabular}
	
    \begin{tabular} {C}
    	\\[-.75em]
    	\textit{\textbf{Computationally inequivalent summations}} \\
    	\\[-.75em]
	\end{tabular}

	\begin{tabular} {T} 
		\hline
		No. of computationally inequivalent summations on $n$ elements having a particular summation SD-tree & Cor. \ref{equiv_sum} & & \\
		\hline
		No. of computationally inequivalent summations on $n$ summands & Prop. \ref{ineq_sum} & Double factorial of odd numbers & \seqnum{A001147} \\
		\hline
		No. of computationally inequivalent ladder, or serial, summations on $n$ summands& Prop. \ref{numladder} & No. of even permutations on $n$ elements & \seqnum{A001710} \\
		\hline
		No. of computationally inequivalent pairwise summations on $n$ summands & Prop. \ref{pairrecursiveold}, Cor. \ref{pairrecursivenew}, \ref{pairclosednew} & No. of tournaments on $n$ teams &  \seqnum{A096351} \\
		\hline
	\end{tabular}

    \begin{tabular} {C}
	    \\[-.75em]
    	\textit{\textbf{$S$-nodes and bounds}} \\
	    \\[-.75em]
\end{tabular}

\begin{tabular} {T} 
	\hline
	Lower bound for the number of $S$-nodes in a summation tree on $n$ summands & Prop. \ref{upper_bound} &The constant 1 & \\
	\hline
	Upper bound for the number of $S$-nodes in a summation tree on $n$ summands & Prop. \ref{lower_bound} & Largest number $k$ where $2^k$ divides $n!$ & \seqnum{A011371} \\
	\hline
    No. of $S$-nodes in a pairwise summation tree with $n$ leaf nodes & Prop.\newline \ref{exp_recursive}, \ref{exp_closed} & Cumulative deficient binary digit sum & \seqnum{A268289} \\
	\hline
	Lower bound for the number of computationally inequivalent summations on $n$ summands on a given SD-tree & Prop. \ref{lower_bound} & Largest odd divisor of $n!$ & \seqnum{A049606} \\
	\hline
	Upper bound for the number of computationally inequivalent summations on $n$ summands on a given SD-tree & Prop. \ref{upper_bound} & No. of even permutations on $n$ elements &  \seqnum{A001710} \\
	\hline
\end{tabular}

\end{table}
\end{document}